\documentclass{IEEEtran}
\usepackage{graphicx}
\usepackage{amsmath,amssymb,amsthm}
\usepackage{mathtools}
\usepackage{float}
\usepackage{xcolor,colortbl}
\usepackage{color}
\usepackage{fp}
\usepackage{tikz}
\usetikzlibrary{shapes.misc,shadows}
\usepackage{multirow}
\usepackage{xcolor,colortbl}
\usepackage{caption}
\usepackage{graphicx}
\usepackage[linesnumbered,ruled]{algorithm2e}
\usepackage{epstopdf}
\usepackage{pgfplots}
\pgfplotsset{compat=1.7}
\usepackage{caption}
\usepackage{subcaption}
\newtheorem{lemma}{Remark}
\newtheorem{proposition}{Lemma}

\usepackage{url}
\usepackage[noadjust]{cite}
\usepackage{setspace}
\DeclareMathOperator{\diag}{diag}
\usepackage{diagbox}
\definecolor{blue}{rgb}{0,0,0}
\usepackage{cuted}
\usepackage{xcolor}

\SetAlFnt{\footnotesize}

\usepackage[abbreviations, automake]{glossaries-extra}
\makeglossaries 
\newabbreviation{OMA}{OMA}{orthogonal multiple access}
\newabbreviation{NOMA}{NOMA}{Non-orthogonal multiple access}
\newabbreviation{RIS}{RIS}{Reconfigurable intelligent surfaces}
\newabbreviation{BS}{BS}{base station}
\newabbreviation{LoS}{LoS}{Line-of-sight}
\newabbreviation{3GPP}{3GPP}{$3^{rd}$ Generation Partnership Project}
\newabbreviation{AoD}{AoD}{angle of departure}
\newabbreviation{SNR}{SNR}{signal-to-noise ratio}
\newabbreviation{ASR}{ASR}{achievable sum rate}
\newabbreviation{PDF}{PDF}{probability distribution function}
\newabbreviation{CDF}{CDF}{cumulative distribution function}
\newabbreviation{5G}{5G}{fifth-generation}
\newabbreviation{LDR}{LDR}{Lagrangian dual reformulation}
\newabbreviation{FPA}{FPA}{fair power allocation}
\newabbreviation{AUP}{AUP}{adaptive user pairing}
\newabbreviation{EE}{EE}{energy efficient}
\newabbreviation{AMO}{AMO}{alternating manifold optimization}
\newabbreviation{SINR}{SINR}{signal-to-interference-plus-noise ratio}

\newacronym{thetaRA}{\ensuremath{\bold{\theta_{\scriptscriptstyle AOA}^{\scriptscriptstyle B,R}}}}{}
\newacronym{thetaRD}{\ensuremath{\bold{\theta_{\scriptscriptstyle AOD}^{\scriptscriptstyle R,k}}}}{}
\newacronym{phiRA}{\ensuremath{\bold{\phi_{\scriptscriptstyle AOA}^{\scriptscriptstyle B,R}}}}{}
\newacronym{phiRD}{\ensuremath{\bold{\phi_{\scriptscriptstyle AOD}^{\scriptscriptstyle R,k}}}}{}

\newacronym{thetaBDR}{\ensuremath{\bold{\theta_{\scriptscriptstyle AOD}^{\scriptscriptstyle B,R}}}}{}
\newacronym{thetaBDk}{\ensuremath{\bold{\theta_{\scriptscriptstyle AOD}^{\scriptscriptstyle B,k}}}}{}
\newacronym{phiBDR}{\ensuremath{\bold{\phi_{\scriptscriptstyle AOD}^{\scriptscriptstyle B,R}}}}{}
\newacronym{phiBDk}{\ensuremath{\bold{\phi_{\scriptscriptstyle AOD}^{\scriptscriptstyle B,k}}}}{}

\newacronym{thetaUAB}{\ensuremath{\bold{\theta_{\scriptscriptstyle AOA}^{\scriptscriptstyle B,k}}}}{}
\newacronym{thetaUAR}{\ensuremath{\bold{\theta_{\scriptscriptstyle AOD}^{\scriptscriptstyle R,k}}}}{}
\newacronym{phiUAB}{\ensuremath{\bold{\phi_{\scriptscriptstyle AOA}^{\scriptscriptstyle B,k}}}}{}
\newacronym{phiUAR}{\ensuremath{\bold{\phi_{\scriptscriptstyle AOD}^{\scriptscriptstyle R,k}}}}{}
\newacronym{aB}{\ensuremath{\bold{a_{\scriptscriptstyle B}}}}{}
\newacronym{aR}{\ensuremath{\bold{a_{\scriptscriptstyle R}}}}{}
\newacronym{au}{\ensuremath{\bold{a_{\scriptscriptstyle k}}}}{}
\newacronym{rB}{\ensuremath{\bold{r_{\scriptscriptstyle B}}}}{}
\newacronym{rR}{\ensuremath{\bold{r_{\scriptscriptstyle R}}}}{}
\newacronym{rbr}{\ensuremath{\bold{r_{\scriptscriptstyle B,R}}}}{}
\newacronym{rbk}{\ensuremath{\bold{r_{\scriptscriptstyle B,k}}}}{}
\newacronym{rrk}{\ensuremath{\bold{r_{\scriptscriptstyle R,k}}}}{}
\newacronym{rrb}{\ensuremath{\bold{r_{\scriptscriptstyle R,B}}}}{}
\newacronym{rkb}{\ensuremath{\bold{r_{\scriptscriptstyle k,B}}}}{}
\newacronym{rkr}{\ensuremath{\bold{r_{\scriptscriptstyle k,R}}}}{}

\newacronym{dbr}{\ensuremath{\bold{d_{\scriptscriptstyle B,R}}}}{}
\newacronym{dbk}{\ensuremath{\bold{d_{\scriptscriptstyle B,k}}}}{}
\newacronym{drk}{\ensuremath{\bold{d_{\scriptscriptstyle R,k}}}}{}
\newacronym{drb}{\ensuremath{\bold{d_{\scriptscriptstyle R,B}}}}{}
\newacronym{dkb}{\ensuremath{\bold{d_{\scriptscriptstyle k,B}}}}{}
\newacronym{dkr}{\ensuremath{\bold{d_{\scriptscriptstyle k,R}}}}{}

\newacronym{dB}{\ensuremath{\bold{d_{BS}}}}{}
\newacronym{dR}{\ensuremath{\bold{d_{RIS}}}}{}
\newacronym{du}{\ensuremath{\bold{d_{k}}}}{}
\newacronym{M}{\ensuremath{M}}{Number of antennae at the base station}
\newacronym{N}{\ensuremath{N}}{Number of antennae at the RIS}
\newacronym{U}{\ensuremath{U}}{Number of antennae at the user}
\newacronym{G}{\ensuremath{\bold{G}}}{Channel between BS to RIS}
\newacronym{F}{\ensuremath{\bold{F}_{k}^H}}{Channel between RIS and user $i$}
\newacronym{H}{\ensuremath{\bold{H}_{k}^H}}{Number of antennae at the IRS}
\newacronym{h}{\ensuremath{\bold{h}_k^H}}{Channel between IRS and user $2$}
\newacronym{Theta}{\ensuremath{\bold{\Theta}}}{Signal strength of the link between BS to IRS}
\newacronym{ThetaO}{\ensuremath{\bold{\Theta}}_\text{opt}}{Signal strength of the link between BS to IRS}
\newacronym{ThetaD}{\ensuremath{\bold{\Theta}}^\text{D}_\text{opt}}{Signal strength of the link between BS to IRS}
\newacronym{ThetaE}{\ensuremath{\bold{\Theta}}_\text{e}}{Signal strength of the link between BS to IRS}
\newacronym{Fst}{\ensuremath{\mathcal{F}}}{Signal strength of the link between IRS to the user $i$}
\newacronym{gamma}{\ensuremath{{\Gamma}_k}}{Channel between IRS and user $1$}
\newacronym{gammaOR}{\ensuremath{{\Gamma}_k}^\text{Only RIS}}{Channel between IRS and user $1$}
\newacronym{gammaOpt}{\ensuremath{{\Gamma}_k}^\text{Opt}}{Channel between IRS and user $1$}
\newacronym{gammaNR}{\ensuremath{{\Gamma}_k^\text{No RIS}}}{Channel between IRS and user $1$}
\newacronym{Rasr}{\ensuremath{R_\text{sum}}}{Signal strength of the link between IRS to the user $1$}
\newacronym{Rasro}{\ensuremath{R_\text{sum}^\text{\tiny OMA}}}{Signal strength of the link between IRS to the user $1$}
\newacronym{gammao}{\ensuremath{{\Gamma}_k^\text{\tiny OMA}}}{Signal strength of the link between IRS to the user $1$}
\newacronym{Pmax}{\ensuremath{P_t}}{Signal strength of the link between IRS to the user $2$}
\newacronym{deltaMSD}{\ensuremath{\Delta_{k,k-1}^\text{\tiny MSD}}}{Azimuth angle of arrival at the IRS}
\newacronym{DepIa}{\ensuremath{\psi_I^a}}{Azimuth angle of departure from the IRS}
\newacronym{ArrIe}{\ensuremath{\phi_I^e}}{Zenith angle of arrival at the IRS}
\newacronym{DepIe}{\ensuremath{\psi_I^e}}{Zenith angle of departure from the IRS}
\newacronym{DepBa}{\ensuremath{\psi_B^a}}{Azimuth angle of departure from the BS}
\newacronym{DepBe}{\ensuremath{\psi_B^e}}{Zenith angle of departure from the BS}
\newacronym{B}{\ensuremath{B}}{Number of beams}
\newacronym{ai}{\ensuremath{\alpha_i}}{Power allocated to strong user}
\newacronym{a1F}{\ensuremath{\alpha_1^\text{\tiny FPA}}}{Power allocated to strong user}
\newacronym{a1M}{\ensuremath{\alpha_1^\text{\tiny MPA}}}{Power allocated to strong user}
\newacronym{a2}{\ensuremath{\alpha_2}}{Power allocated to weak user}
\newacronym{Pt}{\ensuremath{P_t}}{Total available transmit power at the BS}
\newacronym{DiagI}{\ensuremath{\bold{\Theta}}}{IRS reflection diagonal matrix}
\newacronym{DiagI'}{\ensuremath{\bold{\widetilde{\Theta}}}}{IRS reflection diagonal matrix with imperfect phase compensation}
\newacronym{delta}{\ensuremath{\delta}}{Maximum possible phase noise}
\newacronym{GammaW}{\ensuremath{\gamma_2^\text{\tiny NOMA}}}{SINR of weak user in a NOMA system}
\newacronym{GammaS}{\ensuremath{\gamma_1^\text{\tiny NOMA}}}{SINR of strong user in a NOMA system}
\newacronym{GammaI}{\ensuremath{\gamma_i^\text{\tiny NOMA}}}{SINR of $i^{th}$ user in a NOMA system}
\newacronym{RateW}{\ensuremath{R_2^\text{\tiny NOMA}}}{SINR of weak user in a NOMA system}
\newacronym{RateS}{\ensuremath{R_1^\text{\tiny NOMA}}}{SINR of strong user in a NOMA system}
\newacronym{RateI}{\ensuremath{R_i^\text{\tiny NOMA}}}{SINR of $i^{th}$ user in a NOMA system}
\newacronym{RateNoma}{\ensuremath{R_i^\text{\tiny NOMA}}}{SINR of weak user in a NOMA system}
\newacronym{RateO}{\ensuremath{R_i^\text{\tiny OMA}}}{SINR of weak user in a NOMA system}
\newacronym{RateO1}{\ensuremath{\overline{R}_1^\text{\scriptsize }}}{SINR of weak user in a NOMA system}
\newacronym{RateO2}{\ensuremath{\overline{R}_2^\text{\scriptsize }}}{SINR of weak user in a NOMA system}
\newacronym{RateOi}{\ensuremath{\overline{R}_i^\text{\scriptsize }}}{SINR of weak user in a NOMA system}
\newacronym{sigma2}{\ensuremath{\sigma^2}}{Noise power}
\newacronym{I}{\ensuremath{I}}{Interference}
\newacronym{GammaCSI}{\ensuremath{\gamma_i^{\text{\tiny CSI}}}}{SINR of strong user in a OMA system}
\newacronym{Gamma1CSI}{\ensuremath{\gamma_1^\text{\tiny CSI}}}{SINR of weak user in a OMA system}
\newacronym{Gamma2CSI}{\ensuremath{\gamma_2^\text{\tiny CSI}}}{SINR of weak user in a OMA system}
\newacronym{GammaIO}{\ensuremath{\gamma_i^\text{\tiny OMA}}}{SINR of $i^{th}$ user  in a OMA system}
\newacronym{Ow}{\ensuremath{\mathcal{O}_2^\text{\tiny }}}{SINR of $i^{th}$ user  in a OMA system}
\newacronym{Os}{\ensuremath{\mathcal{O}_1^\text{\tiny }}}{SINR of $i^{th}$ user  in a OMA system}
\newacronym{Oi}{\ensuremath{\mathcal{O}_i^\text{\tiny }}}{SINR of $i^{th}$ user  in a OMA system}
\newacronym{thetak'}{\ensuremath{\hat{\theta}_n}}{Maximum possible phase noise}
\DeclareMathOperator{\sinc}{sinc}
\newcommand\copyrighttext{%
  \footnotesize \textcopyright This work has been submitted to the IEEE for possible publication. Copyright may be transferred without notice, after which this version may no longer be accessible.}
\newcommand\copyrightnotice{%
\begin{tikzpicture}[remember picture,overlay]
\node[anchor=south,yshift=10pt] at (current page.south) {\fbox{\parbox{\dimexpr\textwidth-\fboxsep-\fboxrule\relax}{\copyrighttext}}};
\end{tikzpicture}%
}
\begin{document}
\bstctlcite{IEEEexample:BSTcontrol}
\title{\huge Optimizing the Placement and Beamforming of RIS in Cellular Networks: A System-Level Modeling Perspective }
\author{\IEEEauthorblockN{Pavan Reddy M ., SaiDhiraj Amuru, and Kiran Kuchi \vspace{-0.6cm}}
\thanks{ Pavan Reddy M. is with WiSig Networks, Hyderabad, India. Kiran Kuchi and SaiDhiraj Amuru are with the Department of Electrical Engineering, Indian Institute of Technology Hyderabad, Telangana, India.  \newline {(e-mail: pavan@wisig.com,~asaidhiraj@ee.iith.ac.in,~kkuchi@ee.iith.ac.in).}}
}
\maketitle
\copyrightnotice
\begin{abstract}
In this letter, we present in detail the system-level modeling of reconfigurable intelligent surface (RIS)-assisted cellular systems by considering a 3-dimensional channel model between base station, RIS, and user. We prove that the optimal placement of RIS to achieve wider coverage is exactly opposite to the base station, under the constraint of single RIS in each sector. We propose a novel beamforming design for RIS-assisted cellular systems  and derive the achievable sum rate in the presence of ideal, discrete, and random phase shifters at RIS. Through extensive system-level evaluations, we then show that the proposed beamforming design achieves significant improvements as compared to the state-of-the-art algorithms.
\end{abstract}

\IEEEpeerreviewmaketitle 
\section{Introduction}
\gls{RIS} is  considered a promising technology for the next generation cellular communications to improve the achievable network capacity and cellular coverage~\cite{Chandra,Self1}. \gls{RIS} consists of a large number of passive antenna elements (or meta-surfaces) which can reflect the incident ray toward the desired direction. By controlling the impedance of the  meta-surfaces through a passive electronic circuit, an additional phase shift is introduced into the reflected signal, and thus, the signal is steered in the desired direction~\cite{Coverage,Self1,MISO,CosTheta}. This way, \gls{RIS} helps in achieving improved signal reception for the users and also providing coverage to the users who are affected by the signal blockages. 

In the existing literature, \gls{RIS} has been extensively analyzed and significant improvements are demonstrated for single cell scenario~\cite{RISNOMA,IRSNOMACom,IRSNOMATcom}.  In~\cite{RISNOMA}, the authors propose a joint beamforming algorithm to maximize the network capacity.
In~\cite{IRSNOMACom,IRSNOMATcom}, the authors propose power efficient and sum-rate maximizing algorithms for \gls{RIS}-assisted systems. In~\cite{Initial}, the authors present an initial access protocol for \gls{RIS} aided cellular systems. However, these works do not consider the multi-cell analysis while quantifying the achievable gains. Note that only in the multi-cell analysis, the impact of the inter-cell interference from the beamformed \gls{RIS} is captured, and thus, realistically achievable gains can be understood. Additionally, while performing such multi-cell analysis, a 3-dimension channel modeling between the \gls{BS}, \gls{RIS}, and user has to be considered to quantify the achievable gains. However, very few works in the literature have considered the 3-dimension channel model and carried out system-level evaluations~\cite{Random,SimRis1,SimRis2,SimRis3,36873}.

Further, the beamforming at \gls{RIS} has a significant impact on the achievable gains with \gls{RIS}-assisted systems~\cite{TSINGHUA}.   In~\cite{Random}, the authors consider deploying multiple \gls{RIS} and using random phase shifters at each \gls{RIS}. This is a low-complex way to enhance the network performance but only a few users whose channel coefficients are aligned with those assigned phase shifters will observe spectral enhancements. In~\cite{Chandra}, the authors consider random phase shift allocation in each time slot and assume that the users decode the pilot signals transmitted at the beginning of each slot.
Then, the users' feedback the channel quality reports to the \gls{BS}, and the users with the best instantaneous channel conditions will be scheduled for data transmission within the rest of the slot. With this approach, the effective channel observed by each user changes from slot to slot and \gls{BS} schedules the data transmissions according to the channel quality reports to maximize the achievable capacities. However, this procedure requires a large number of active users to realize the desired gains, and also, the feedback from the users within the same time slots is difficult to realize in practice. Hence, there is a need to consider all the aforementioned details and design a  practically feasible, low-complex, and yet optimal way of  beamforming at \gls{RIS}. 

\begin{table*}
\small
\begin{align}
\gls{H}&=\left[\gls{au}\left(\gls{thetaUAB},\gls{phiUAB}\right) \right]^T.\  \bold{\chi_{\scriptscriptstyle B,k}}. \ \gls{aB}\left(\gls{thetaBDk},\gls{phiBDk}\right). \exp\left(\dfrac{j2\pi(\gls{rbk}\gls{dbk})}{\lambda}\right).
\exp\left(\dfrac{j2\pi(\gls{rkb}\gls{dkb})}{\lambda}\right).
\exp\left(\dfrac{j2\pi\gls{rbk}}{\lambda}\bold{\nu_{B,k}} t\right)\label{eqn:HFD1}\\
\gls{G}&=\left[\gls{aR}\left(\gls{thetaRA},\gls{phiRA}\right) \right]^T.\  \bold{\chi_{\scriptscriptstyle B,R}}. \ \gls{aB}\left(\gls{thetaBDR},\gls{phiBDR}\right). \exp\left(\dfrac{j2\pi(\gls{rbr}\gls{dbr})}{\lambda}\right).
\exp\left(\dfrac{j2\pi(\gls{rrb}\gls{drb})}{\lambda}\right)\label{eqn:HFDH3}\\
\gls{F}&=\left[\gls{au}\left(\gls{thetaUAR},\gls{phiUAR}\right) \right]^T  \hspace{-0.2cm}. \ \bold{\chi_{\scriptscriptstyle R,k}}. \ \gls{aR}\left(\gls{thetaRD},\gls{phiRD}\right). \exp\left(\dfrac{j2\pi(\gls{rrk}\gls{drk})}{\lambda}\right).
\exp\left(\dfrac{j2\pi(\gls{rkr}\gls{dkr})}{\lambda}\right).
\exp\left(\dfrac{j2\pi\gls{rrk}}{\lambda}\bold{\nu_{R,k}} t\right) \label{eqn:H3}
\end{align}
\vspace{-0.6cm}
\end{table*}
Motivated by these facts, we present the following key contributions in this letter.
{
\begin{itemize}

\item We prove that optimal \gls{RIS} placement (under the constraint of single \gls{RIS} in each sector) to achieve a wider coverage is exactly opposite to the base station with boresight of \gls{RIS} facing the boresight of the base station.
\item We propose a novel beamforming design for \gls{RIS}-assisted cellular systems and  analyze the impact of the selection of phase shifters on the proposed design. We derive the achievable sum rates in the presence of ideal, discrete, and random phase shifters.
\item We present the system-level modeling of \gls{RIS}-assisted cellular systems by 
considering the 3-dimensional channel modeling between the \gls{BS}, \gls{RIS}, and 
user. We perform extensive system-level simulations and show that the proposed beamforming design significantly outperforms the state-of-the-art algorithms.  
\end{itemize}}

\section{System Model}
\label{sec:SysModel}
\subsection{Scenario Description}
We consider a cellular network with \gls{BS}, \gls{RIS}, and users as shown in Fig.~\ref{fig:SysModel}. We assume \gls{M}, \gls{N}, and \gls{U} as the number of antenna elements at the \gls{BS}, \gls{RIS}, and user, respectively, and formulate the equivalent channel $\gls{h}\in \mathbb{C}^{U \times M}$ between the \gls{BS} and user $k$ as follows.
\begin{align}
\gls{h}=\gls{H}+\gls{F}\gls{Theta}\gls{G},
\label{eqn:heq}
\end{align}
 where, $[\cdot]^H$ represents Hermitian operation, and $\gls{G} \in \mathbb{C}^{N \times M}$, $\bold{F}\in \mathbb{C}^{U \times N}$, and $\bold{H}\in \mathbb{C}^{U \times M}$ denote the channel coefficients between the \gls{BS} to \gls{RIS}, \gls{RIS} to the user $k$, and the direct link between \gls{BS} and user $k$, respectively. Further, $\gls{Theta}\in \mathbb{C}^{N \times N}$ denotes the phase shift matrix at the \gls{RIS} formulated as follows~\cite{RISNOMA,Coverage}
\begin{align}
\gls{Theta}=&\ \diag(e^{j\theta_{1}},e^{j\theta_{2}}, \ldots,e^{j\theta_{N}}), \forall \ \theta_{n} \in \gls{Fst},\nonumber\\
\gls{Fst}\triangleq&\ \Big\{\theta_{n}\Bigl\vert |e^{j\theta_{n}}|\leq 1 \Big\}, \forall\ 1\leq n\leq N.\nonumber
\end{align}
\begin{figure}
    \hspace{-0.3cm}
    \begin{subfigure}{0.2\textwidth}
    \includegraphics[scale=0.35]{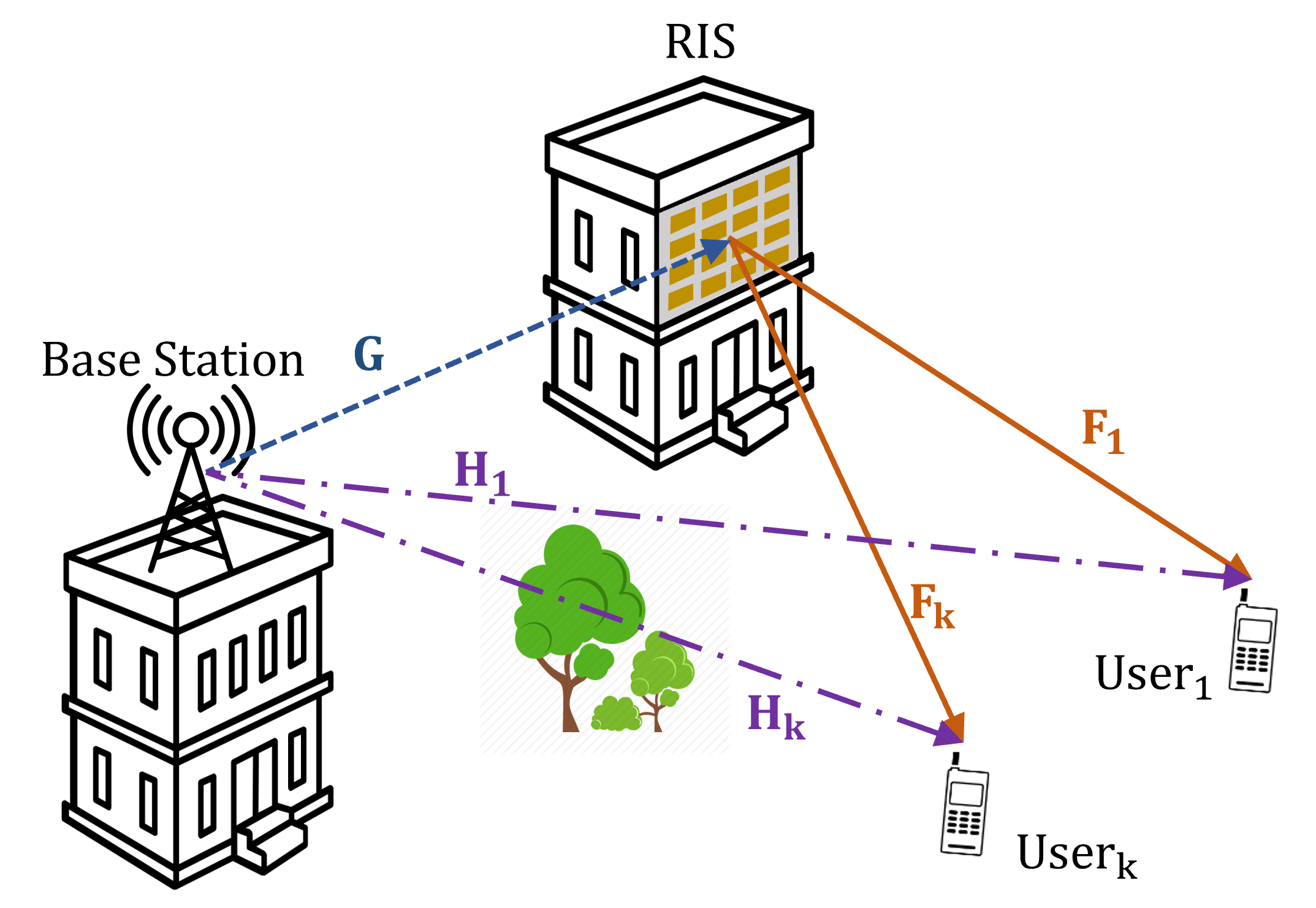}
    \caption{Scenario Depiction}
    \end{subfigure}
    \hspace{2cm}
    \begin{subfigure}{0.2\textwidth}
    \includegraphics[scale=0.35]{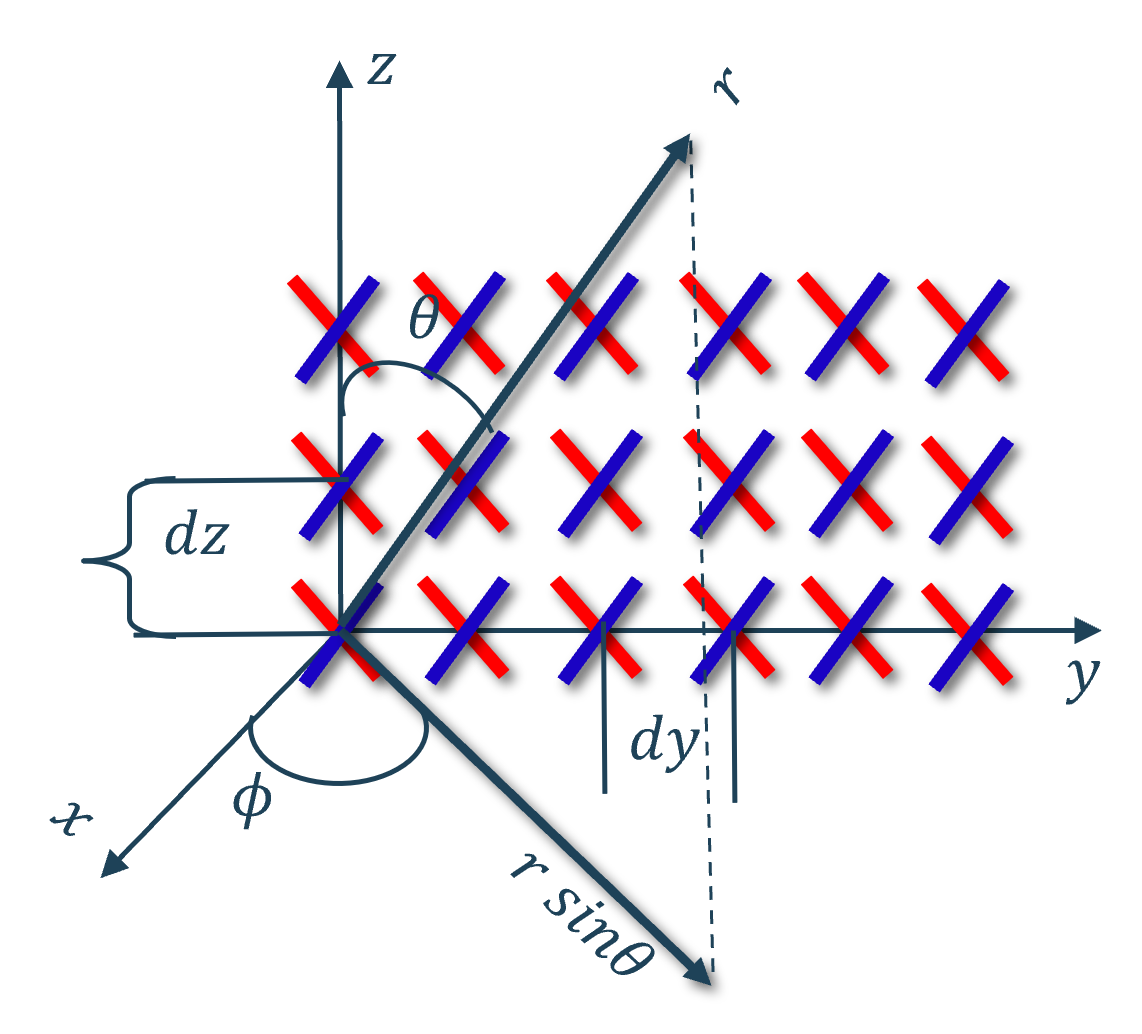}
    \caption{RIS Modelling}
    \end{subfigure}
    \caption{System Model.}
    \label{fig:SysModel}
\vspace{-.5cm}
\end{figure}
Note that ideal phase shifters are difficult to realize in practice. Thus, we formulate the discrete phase shifters as follows~\cite{TSINGHUA} 
\begin{align}
\mathcal{F}_D\triangleq \bigg\{\theta_{n} \Bigl \vert \theta_{n} \in \Big\{0, \tfrac{2\pi}{D}, \ldots, \tfrac{2\pi (D-1)}{D}\Big\} \bigg\}, \label{eqn:FD}
\end{align}
where $D$ denotes the number of discrete-levels at \gls{RIS}.
We consider the \gls{RIS}-assisted downlink transmission for the user $k$ as follows.
\begin{align}
\bold{x}_k&=\sqrt{\gls{Pmax}} s_k,\nonumber
\end{align}
where $s_k$ is the transmitted symbol corresponding to the user $k$ with unit power i.e., $\mathbb{E}\{s_k s_k^H\}=1$ and \gls{Pmax} is the maximum power the \gls{BS} can transmit. 
Thus, we formulate the received signal at the user $k$ as follows.
\begin{align}
\bold{y}_k=&\gls{h} \bold{x}_k+\bold{z}_k,\nonumber\\
=& \Big(\gls{H}+ \bold{F}_k^H\bold{\Theta}\bold{G}\Big)\bold{x}_k+\bold{z}_k.\label{eqn:hsmpl}\end{align}\\where  $\bold{z}_k$ represents the additive white Gaussian noise at the user $k$ with zero mean and co-variance $\bold{\Xi}_k=\sigma^2\bold{I}$. 

\subsection{3-dimensional Channel Model}
Typically, the physical size of the \gls{RIS} structure is comparatively much smaller than the actual distance of the \gls{RIS} from the \gls{BS} and user. Hence, we assume the \textit{far-field region} while formulating the channel model~\cite{Initial}. Further, considering the various existing evaluations in the literature~\cite{SimRis1, SimRis2,SimRis3}, we adapt the \gls{3GPP} based 3-dimension channel between the \gls{BS}, \gls{RIS}, and user~\cite{38901,SimRis1}. We define the clusters, sub-clusters, rays,  angle of arrivals, angle of departures, path loss, shadow-fading, and small-scale fading as per the channel modeling formulated in \cite{38901}. The formulation of channel coefficients between the \gls{BS}, \gls{RIS}, and user (excluding the large scale fading parameters such as path loss and shadow fading) is summarized in \eqref{eqn:HFD1}-\eqref{eqn:H3}, where subscripts $\scriptstyle \bold{B}$, $\scriptstyle \bold{R} $, and $\scriptstyle \bold{k}$ represent base stations, \gls{RIS}, and user, respectively. 
$\bold{\phi_{AOA}}, \ \bold{\theta_{AOA}},\   \bold{\phi_{AOD}},$ and $\bold{\theta_{AOD}}$ represent the angle of arrivals in azimuth and elevation, and angle of departures in azimuth and elevation, respectively. Based on the cellular environment (such as urban, rural, etc.), these angles are assigned considering various sub-rays corresponding to various clusters between the transmitter and receiver as per~\cite{36873,38901}.
\gls{au}, \gls{aR}, and \gls{aB} represent the antenna gains at the user, \gls{RIS}, and \gls{BS}, respectively. We consider the sectoral antenna pattern as per~\cite{38901,36873} at the \gls{BS} and omni directional antenna at the user. We model the transmissions from \gls{RIS} as just reflections (passive) i.e., with zero additional antenna gain ($\gls{aR}(\cdot)\leq1$). $\bold{\chi}$ captures the cross polarisation powers, $\bold{d}$ represents the location vector, and $\bold{r}$ represents the spherical unit vector for corresponding antenna elements, respectively. For the planar array shown in Fig.~\ref{fig:SysModel}, the spherical unit vector is formulated as shown below~\cite{38901}.
\begin{align}
\bold{r}=\begin{bmatrix} \sin \boldsymbol{\theta} \cos \boldsymbol{\phi} \\ \sin \boldsymbol{\theta} \sin \boldsymbol{\phi}\\ \cos \boldsymbol{\theta} \end{bmatrix}.\label{eqn:rm}
\end{align}
$\nu$ represents the speed of the receiver and the corresponding exponential in \eqref{eqn:HFD1} and \eqref{eqn:H3} captures the Doppler effect. Note that no Doppler term is considered in \eqref{eqn:HFDH3}, as we consider \gls{RIS} is stationary. Further, we also consider the path loss and shadow fading for each cellular environment as per~\cite{36873}.

\begin{figure}
\centering
\includegraphics[scale=0.285]{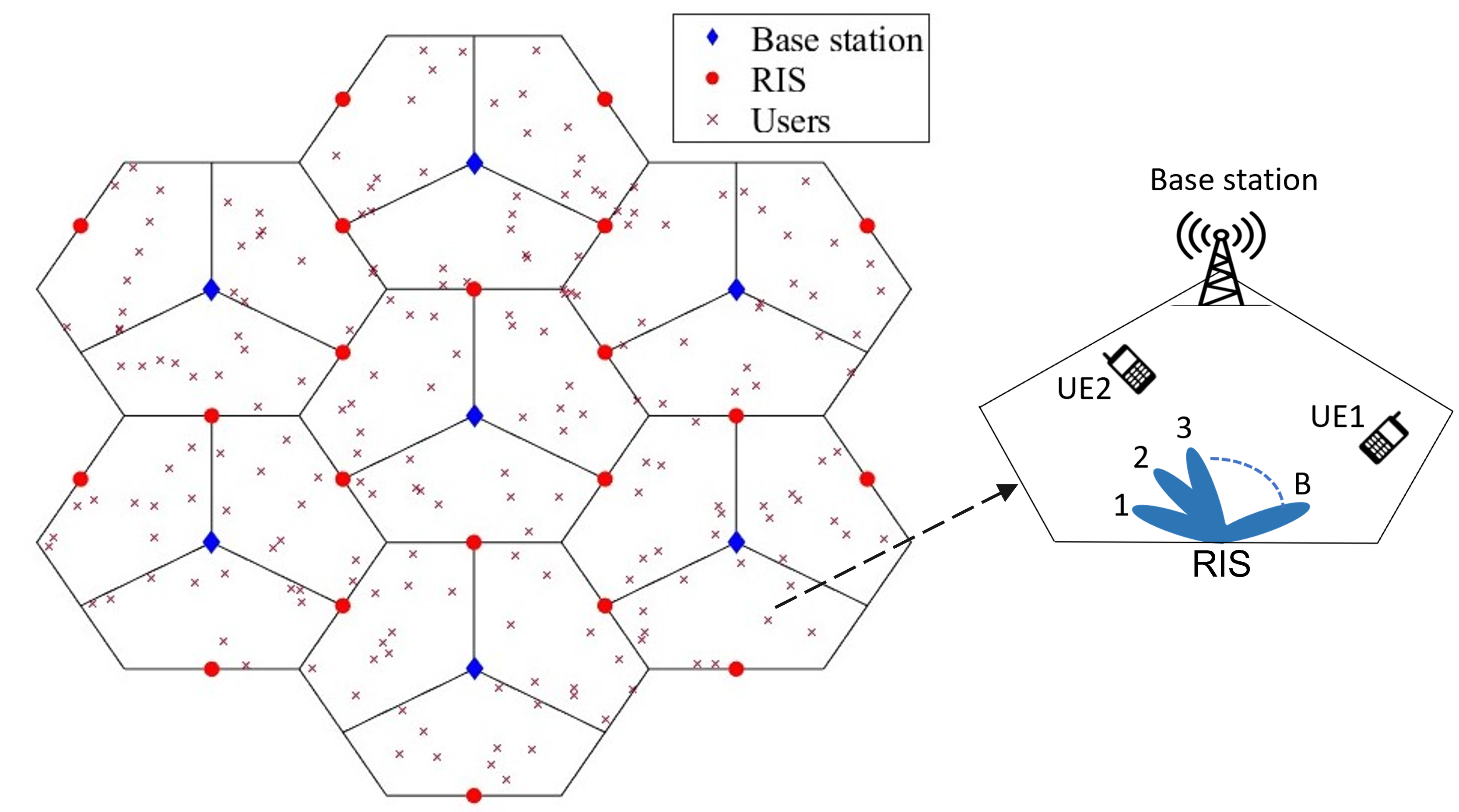}
\caption{Proposed \gls{RIS} organization and beamforming}
\label{fig:RISorg}
\vspace{-0.5cm}
\end{figure}

\section{Proposed \gls{RIS}  Organization and Beamforming}
\label{sec:Proposed}
\subsection{\gls{RIS} organization}
For the \gls{RIS}-assisted cellular networks, we propose the \gls{RIS} organization as shown in Fig.~\ref{fig:RISorg}.  The boresights of the sectors are considered to be in the directions of $[30^\circ, 150^\circ, 270^\circ]$ and the \gls{RIS} is placed at the edge of each sector with the boresight of \gls{RIS} elements facing the boresight of the \gls{BS}. This organization of the \gls{BS} orientations\ and positioning of the \gls{RIS} helps in ensuring \gls{RIS} access to all the users in the sector. When compared to the scenario with boresight of the sectors in the direction of $[0^\circ, 120^\circ, 240^\circ ]$ and \gls{RIS} placed at cell-edge with its elements facing the boresight of the \gls{BS},   the proposed organization has two key benefits as follows. With the proposed organization, the \gls{BS} to \gls{RIS} distance is $\frac{ISD}{2}$ instead of $\frac{ISD}{\sqrt{3}}$, where, $ISD$ indicates the inter-site distance. Note that the achievable gains with \gls{RIS}-assisted systems are heavily dependent on the path loss observed in the \gls{BS}-\gls{RIS}-user link. Thus, with the proposed organization, the path loss observed will be minimal, and hence, enhanced signal reception can be achieved by the user. Further, with the proposed organization, the interference caused by the beamforming of the \gls{RIS} to the other sectors is also comparatively minimum, as the nearest neighbor sectors are not in the boresight direction of the \gls{RIS}.

\begin{proposition}
\label{lemma:LemmaP}
The optimal location for the \gls{RIS} to achieve a wider coverage is at the opposite end of the sector, with the boresight of the \gls{RIS} facing the boresight of the \gls{BS} antenna.
\end{proposition}
\begin{figure}
\centering 
\includegraphics[scale=0.6]{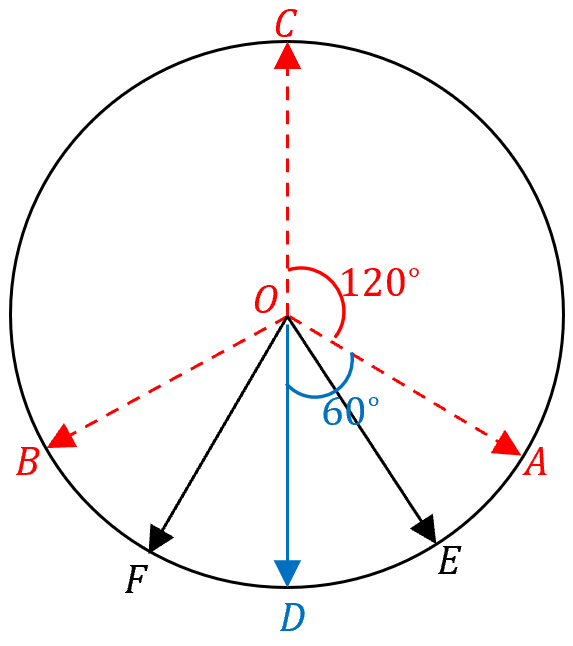}
\caption{Optimal \gls{RIS} location analysis for proof of Lemma~\ref{lemma:LemmaP}. }
\label{fig:Theorem}
\vspace{-0.5cm}
\end{figure}
\begin{proof}
We consider the circular approximation for a cell site and assume the 3-sector division as shown in Fig.~\ref{fig:Theorem}~\cite{Andrea}, where $OA, OB,$ and $OC$ are the sector divisions. We prove that $D$ is the optimal location of the \gls{RIS} for the sector $OAB$, and the same can be extended to the other sectors. Note that the reflections from the \gls{RIS} are observed on only one side of the \gls{RIS}, and hence, to ensure \gls{RIS} access to the entire sector, \gls{RIS} has to be placed only on the circumference of the sector $OAB$. Further, the reflection amplitude of the \gls{RIS} is typically approximated as $\cos(\theta)$, where $\theta$ is the angle of incidence on \gls{RIS}~\cite{CosTheta}. Thus, when \gls{RIS} is placed along the lines $OA$ and $OB$ but facing the sector, the signal from the base station is received at an angle of incidence of $90^\circ$ resulting in minimal gains. Also, note that typical \gls{BS} cellular antenna patterns have the least gains in the direction of sector-edges~\cite{36873}. Thus, the optimal location of the \gls{RIS} can only be on the arc $AB$. 

As shown in Fig.~\ref{fig:Theorem}, consider 3 points on the arc $AB$ which are at $0^\circ, >0^\circ,$ and $<0^\circ$ from the \gls{BS} denoted as $D, E,$ and $F$, respectively. Note that all three points are at the same distance from \gls{BS} (or $O$ in Fig.~\ref{fig:Theorem}). Thus, the distance-dependent path loss observed by the \gls{RIS} at these locations would be similar.
Further, when we place \gls{RIS} at $D$, the maximum distance of any user in the sector from $D$ does not exceed the radius of the sector $r$. For example, $A$ is the farthest location from $D$.  Since, $OA=OD=r$ and $\angle AOD = 60^\circ$, $AD=r$.
However, this is not the case when the \gls{RIS} is placed at either $E$ or $F$. The \gls{RIS} will be closer to some users in the sector, but the farthest locations from $E$ and $F$ are 
$EB>r$ and $FA>r$. This increase in the distance from \gls{RIS} to the users will result in decreased signal reception, and thus provides less coverage to these users when compared with the case of \gls{RIS} at $D$. Hence, the optimal location of \gls{RIS} for wider coverage is $D$. This ends the proof of Lemma~\ref{lemma:LemmaP}. 
\end{proof}

\vspace{-0.7cm}
\subsection{Beamforming design}
To perform optimal beamforming at \gls{RIS} for each user, the \gls{BS} has to know the complete channel state information between \gls{RIS} and each user. However, this is not always achievable in practice. Hence, we propose a practically feasible beamforming design as follows. We configure \gls{B} number of predefined digital beams at the \gls{RIS} such that they cover the entire sector, as shown in Fig.~\ref{fig:RISorg}. The number of beams \gls{B} can be decided by the network operator based on the number of antenna elements at the \gls{RIS}. Given \gls{B} beams, the \gls{BS} can tune the phase shifters at \gls{RIS} and activate one of the configured beams over various time slots. At the receiver, the user calculates the effective received signal strength from the configured beam and also the direct link from the \gls{BS}. The user then reports back the best beam index to the \gls{BS} in the uplink. For example, the best beam for users 1 and 2 in Fig.~\ref{fig:RISorg} are beam-\gls{B} and beam-3, respectively. This way, the \gls{BS} can configure the desired beams at the \gls{RIS} while scheduling the data for each user.  Note that a similar approach can also be adapted for the initial user attach. While transmitting the downlink synchronization signals, the \gls{BS} can tune the phase shifters at \gls{RIS} and transmit different beams over various time slots such that the user can lock to the beam from which it receives maximum signal strength.
\begin{figure*}[t!]
\centering
\includegraphics[scale=0.41]{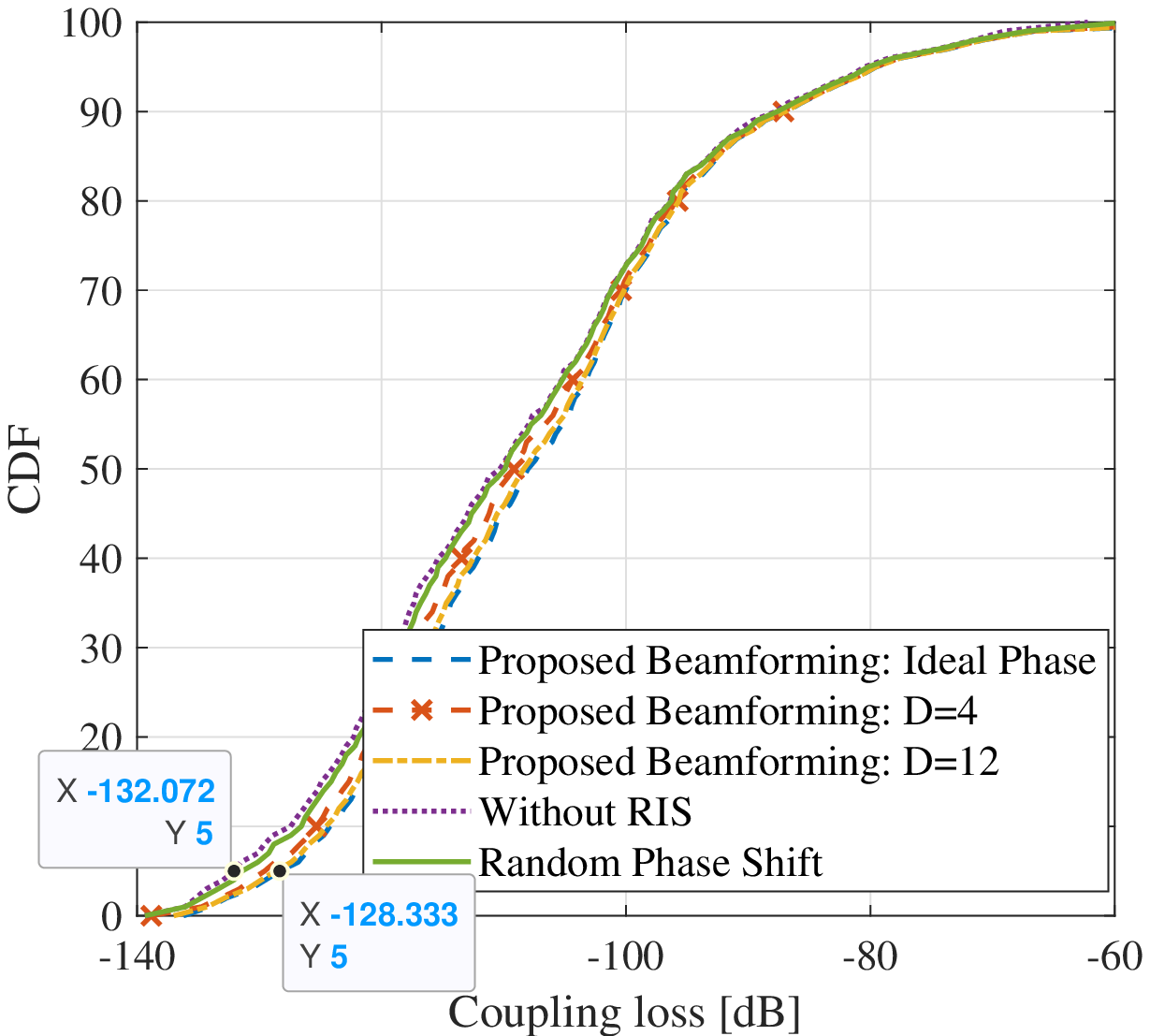}
\includegraphics[scale=0.41]{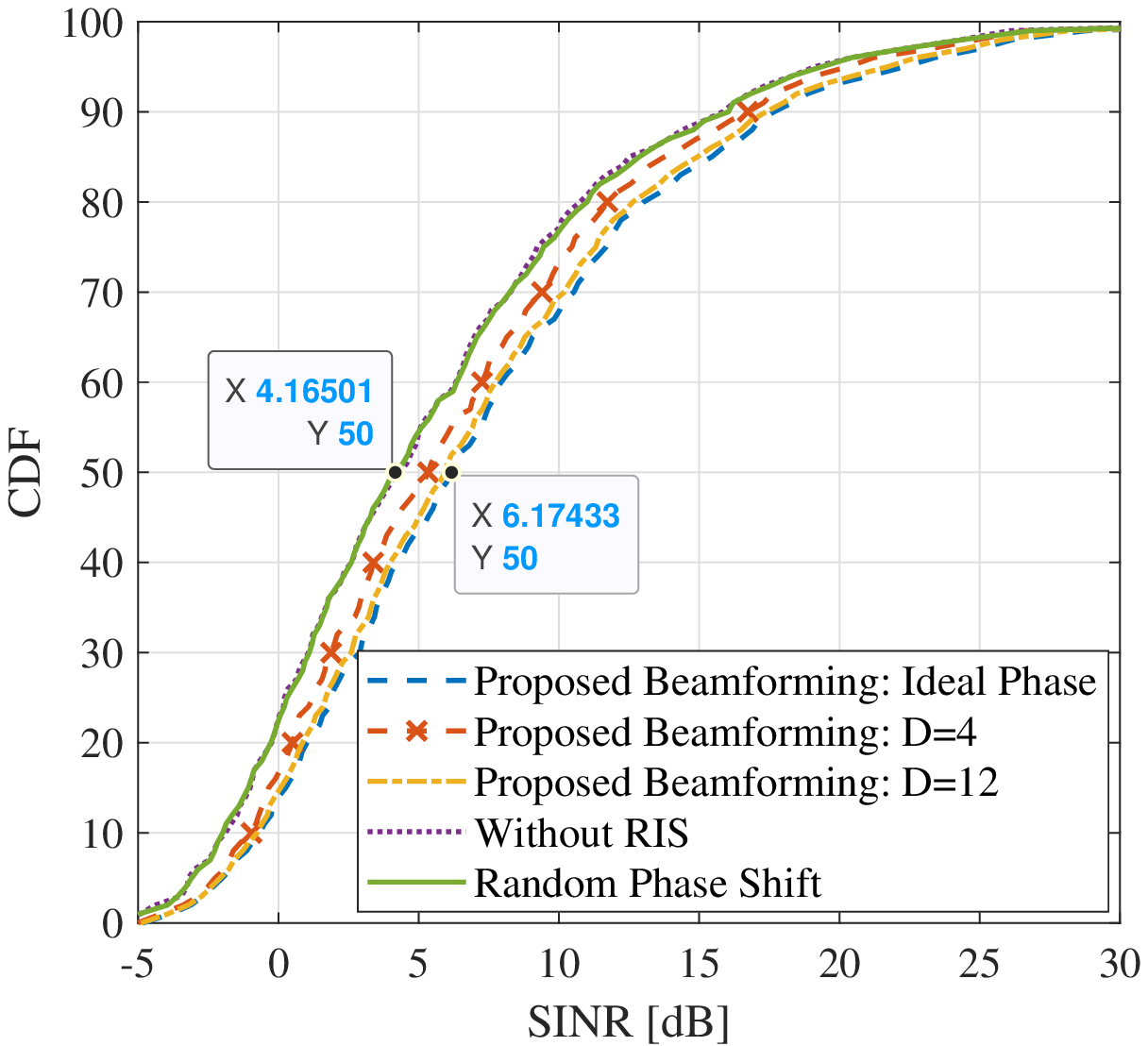}
\includegraphics[scale=0.41]{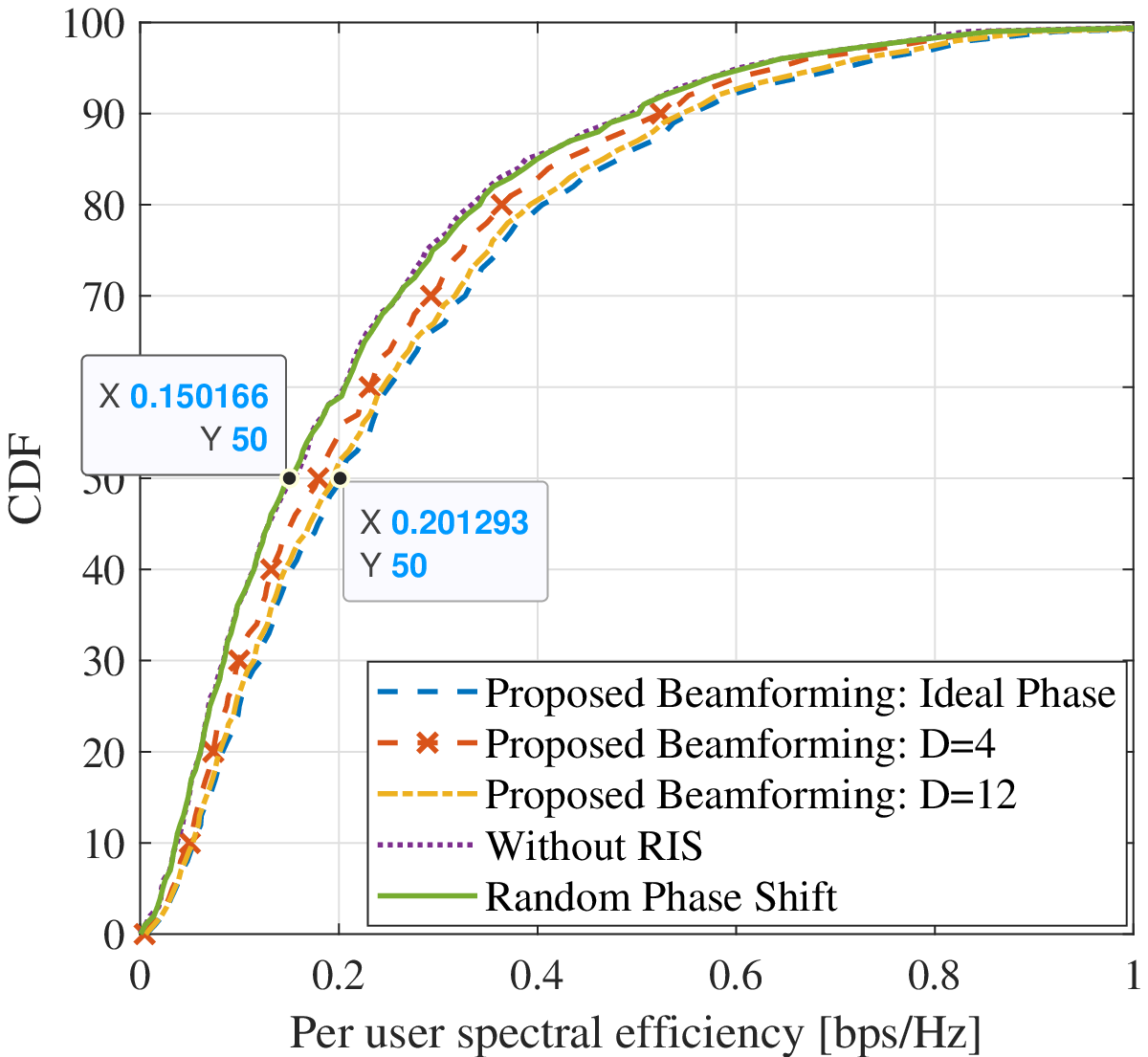}
\caption{System-level evaluation of achievable coupling loss, SINR, and spectral efficiency with various algorithms}
\label{fig:SINR}
\vspace{-0.4cm}
\end{figure*}
\begin{lemma} 
The maximum achievable data rate for user $k$ in the presence of ideal phase shifters at \gls{RIS} is given by
\begin{align}
\log_2\left(1+\gls{Pt}\dfrac{\Big||\gls{H}|+|\gls{F}||\gls{G}|\Big|^2}{\bold{\Xi}_k}\right) \label{eqn:l1}
\end{align}
\end{lemma}
\begin{proof}
The optimal values of phase shifters that maximize the achievable data rate are given by
\begin{align} 
\gls{ThetaO}=\arg \max_{\gls{Theta}} ||\gls{h}||^2 =\arg \max_{\gls{Theta}} ||\gls{H}+\gls{F}\gls{Theta}\gls{G} ||^2\nonumber\end{align}
\begin{align}
\hspace{-1.5cm} \text{From \cite{Chandra}},
 \gls{ThetaO}= \exp\left(j \left(\angle \gls{H}-\angle(\gls{F}+\gls{G})\right)\right) \label{eqn:ThetaOpt}
\end{align}
Using \eqref{eqn:ThetaOpt}, we formulate the achievable \gls{SNR} as
{\small
\begin{align}
\gls{gamma}
=\gls{Pt}\dfrac{\Big|\Big| |\gls{H}|\angle \gls{H}+|\gls{F}| |\gls{G}|\angle \gls{H}\Big|\Big|}{\bold{\Xi}_k}
=\gls{Pt}\dfrac{\Big||\gls{H}|+|\gls{F}||\gls{G}|\Big|^2}{\bold{\Xi}_k}\nonumber
\end{align}
}
Thus, the maximum achievable data rate at user $k$ in the presence of ideal phase shifters at \gls{RIS} is as shown in~\eqref{eqn:l1}.
\end{proof}
\begin{proposition}
As $N\rightarrow\infty$, the maximum achievable data rate at user $k$ with discrete phase shifters at \gls{RIS} is given by
\begin{align} \log_2\left(1+ \gls{Pt}\dfrac{\Big||\gls{H}|+\mathbb{E}[r_i]\sinc(\frac{\pi}{D})\Big|^2}{\bold{\Xi}_k}\right)\end{align}
where $D$ represents the number of discrete-levels of the phase shifters at \gls{RIS} and $r_i$ represents the diagonal elements of $\diag(\gls{F}\gls{ThetaO}\gls{G})$.
\label{lemma:lem2}
\end{proposition}
\begin{proof}
Using the optimal beamforming derived in \eqref{eqn:ThetaOpt} at \gls{RIS} with discrete phase shifters, we get 
\begin{align}
\gls{ThetaD}=\gls{ThetaO}\odot\gls{ThetaE}\label{eqn:thetaD}
\end{align}
where, $\odot$ denotes the element-wise (Hadamard) product and $\gls{ThetaE}=\diag\{e^{j\theta^e_1},\ldots, e^{j\theta^e_N}\}$ captures the quantization error. From~\eqref{eqn:FD}, we obtain $\theta_i^e\in \left[-\tfrac{\pi}{D},\tfrac{\pi}{D} \right]$. Further, using \eqref{eqn:thetaD}, we get
\begin{align} 
\gls{F}\gls{ThetaD}\gls{G}=\gls{F}\left(\gls{ThetaO}\odot\gls{ThetaE}\right)\gls{G}\overset{(a)}{=}(\gls{ThetaE}\bold{1}_N)^H\diag(\gls{F}\gls{ThetaO}\gls{G})\nonumber
\end{align}
where $\bold{1}_N$ is a $N\times 1$ array and (a) is obtained by using Hadamard-product properties. By denoting $r_i$ as the diagonal elements of $\gls{F}\gls{ThetaO}\gls{G}$, we get
\begin{align} 
\gls{F}\gls{ThetaD}\gls{G}&=\sum_{n=1}^{N} e^{j\theta_i^e}r_i\overset{(b)}{=}N\mathbb{E}\big[e^{j\theta_i^e}r_i\big]\nonumber\\
&\overset{(c)}{=} N\mathbb{E}\big[e^{j\theta_i^e}\big] \mathbb{E}\big[r_i\big]\overset{(d)}{=}N\mathbb{E}\big[\cos(\theta^e_i)\big]\mathbb{E}\big[r_i\big]\nonumber\\
&\overset{(e)}{=}N\sinc\left(\tfrac{\pi}{D}\right)\mathbb{E}\big[r_i\big]
\nonumber
\end{align}
(b) is obtained by considering the law of large numbers~\cite{MISO}, (c) is obtained assuming the independent nature of the random variables, (d) is obtained after integration of the odd symmetrical function $\sin(\theta_i^e)$ over $[-\tfrac{\pi}{D}, \tfrac{\pi}{D}]$, and (e) is based on the assumption of the uniform distribution of $\theta_i^e$ over $[-\tfrac{\pi}{D}, \tfrac{\pi}{D}]$
. 
Note that based on \eqref{eqn:ThetaOpt}, $\mathbb{E}\big[r_i\big]$ also has the phase in the direction of $\angle \gls{H}$, and thus, we get
\begin{align}
||\gls{H}+\gls{F}\gls{ThetaD}\gls{G}||^2=\big||\gls{H}|+N\sinc(\tfrac{\pi}{D})\mathbb{E}[r_i]\big|^2 \label{eqn:l2proof}
\end{align}
This completes the proof of Lemma~\ref{lemma:lem2}.
\end{proof}

\begin{proposition}
As $N\rightarrow \infty$, the achievable data rates with the random phase shifters at \gls{RIS} are similar to that of the case with no \gls{RIS}.
\label{lemma:lem3}
\end{proposition}
\begin{proof}
Note that with random phase shifters at \gls{RIS}, $\theta_i^e \in [-\pi, \pi]$, which is same as considering $D=1$.
Thus, substituting $D=1$ in \eqref{eqn:l2proof}, we get 
\begin{align} ||\gls{H}+\gls{F}\gls{ThetaD}\gls{G}||^2=||\gls{H}||^2
\end{align}
This completes the proof of Lemma~\ref{lemma:lem3}.
\end{proof} 
\subsection{Baseline algorithms considered for the evaluation}
\subsubsection{Cellular system without any \gls{RIS}}
We consider cellular system without any \gls{RIS} as the baseline system for the comparison to understand the achievable gains with the introduction of \gls{RIS}. Further, to ensure the correctness of the system-level modelling along with channel formulation, we have considered the antenna configuration and simulation parameters of the \gls{3GPP} Phase-2 Calibration~\cite{36873} as the baseline and aligned the results with~\cite{36873}.

\subsubsection{Random phase shifters at the \gls{RIS}~\cite{Random, Chandra}}
We evaluate the performance of the \gls{RIS}-assisted cellular systems while applying random phase shifters (i.e., values of $\theta_{n}$ are assigned from complex random distribution with $|\theta_n |\leq 1$) at the \gls{RIS} in each time slot~\cite{Random}. 
%
%
\section{Simulation Results and Discussion}
\label{sec:Results}
\begin{table}
\scriptsize
\centering
\caption{Simulation parameters}
\label{tab:simp}
\begin{tabular}{|m{3.3cm}|m{4.4cm}|c|}
\hline
\textbf{\textit{Parameter}}&\textbf{\textit{Values}}\\\hline \hline
Cellular environment & Urban macro (UMa)~\cite{38901}\\ \hline 
Carrier frequency & 2~GHz\\ \hline 
Path loss model & UMa~\cite{38901}\\ \hline
System Bandwidth & 10 MHz \\ \hline
User drop & 10 per sector \\ \hline 
\gls{BS} transmit power $P_t$& 43~dBm \\ \hline
Wraparound & Geographical distance based \\ \hline
\gls{RIS} drop & 1 per sector\\ \hline
User attachment & Based on received signal power \\ \hline 
Small scale model & \gls{3GPP} FD-MIMO~\cite{38901}\\ \hline
\gls{BS} antenna elements, $M$& $10 \times 2 \times 2 = 40$ elements\\ \hline 
\gls{BS} antenna pattern & \gls{3GPP} sectoral pattern\\ \hline
User antenna elements, $U$ & 1 element \\ \hline
User antenna pattern & omni directional\\ \hline
\gls{RIS} antenna elements, $N$ & $16 \times 16 =256$ elements \\ \hline
\gls{RIS} antenna pattern & passive (only reflections modelled) \\ \hline
Thermal Noise & -174 dBm/Hz\\ \hline
Inter-cell interference & explicitly modelled from both \gls{RIS}s and \gls{BS}s\\ \hline
Other general parameters  & as per \cite{36873} \\ \hline
\end{tabular}
\end{table}

The simulation parameters considered for the system-level evaluation of the proposed algorithms are presented in Table~\ref{tab:simp}. 
We have performed Monte-Carlo simulations in MATLAB and presented the achievable coupling loss, \gls{SINR}, and spectral efficiency with all the considered algorithms in Fig.~\ref{fig:SINR}. 
The coupling loss is calculated as the difference between the transmitted power and received power at each user. 
In Fig.~\ref{fig:SINR}, for the case without \gls{RIS}, the coupling loss  is aligned with the \gls{3GPP} Phase 2 Calibration curves presented in \cite{36873, 38901}, thus, ensuring the correctness of our implementation. As compared to the scenario with no \gls{RIS}, the proposed design with ideal phase shifters achieves a 4 dB gain in the received signal power. The gains achievable with random phase shifters while using single \gls{RIS} are very minimal. Further, the proposed design with 4-bit discrete phase shifters also has significant improvement as compared to the scenario without any \gls{RIS}.  With the 12-bit discrete phase shifters, the coupling loss with the proposed design is similar to the coupling loss with ideal phase shifters.

\begin{figure}[t!]
\centering
\includegraphics[scale=0.38]{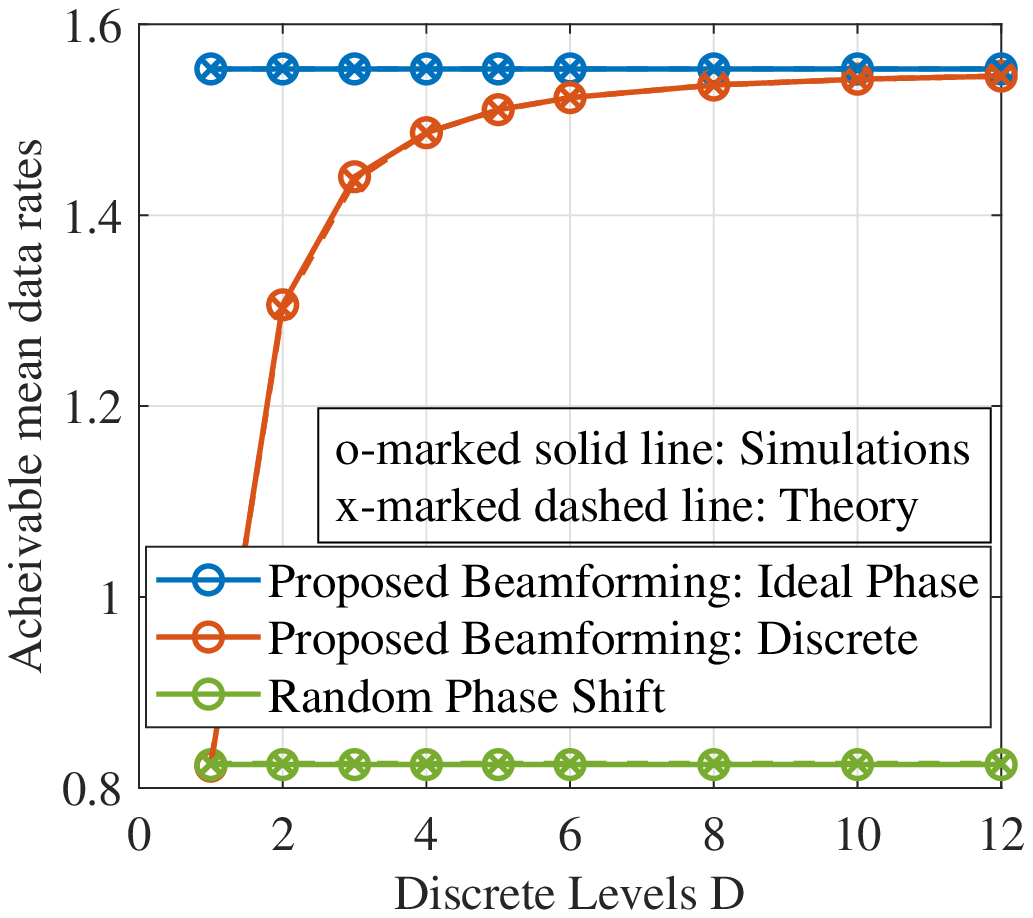}
\includegraphics[scale=0.38]{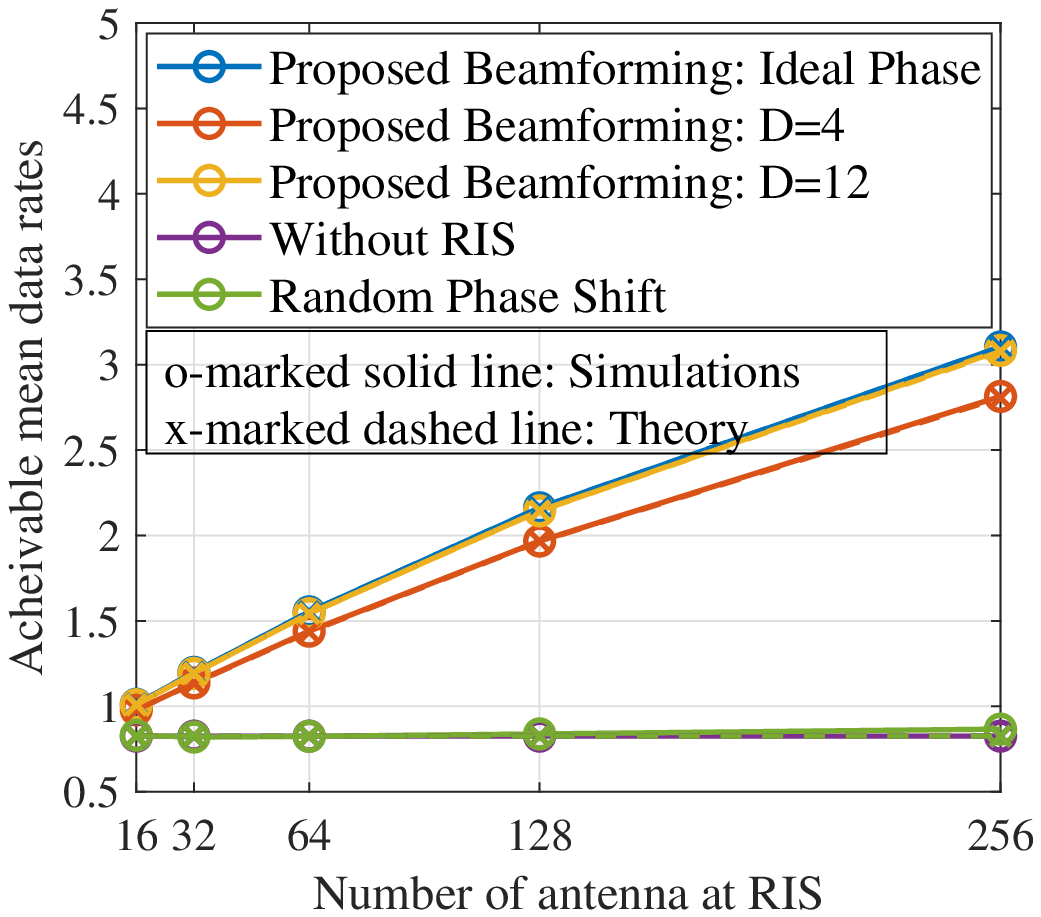}
\caption{Comparison of data rates with varying $D$ and $N$}
\label{fig:SINR}
\end{figure}

We have calculated the \gls{SINR} by explicitly capturing the interference from the neighboring \gls{BS} and \gls{RIS} as follows.
\begin{align}
\text{\gls{SINR}}=\gls{Pmax}\frac{\bold{h}_k\gls{h}}{\text{IN}_k+\bold{\Xi}_k},
\end{align}
where $\text{IN}_k$ represents the interference captured from the neighbouring \gls{BS} and \gls{RIS} transmissions. 
In Fig.~\ref{fig:SINR}, the \gls{SINR} for the scenario without \gls{RIS} is aligned with \gls{3GPP} Phase 2 Calibration curves presented in \cite{36873}, and thus, calibrating the  \gls{SINR} calculation. As compared to the scenario with no \gls{RIS}, the proposed design with ideal phase shifters achieves close to 2~dB improvement in the \gls{SINR}. Note that despite the maximum of 4~dB gain in the coupling loss, we achieve only 2~dB improvements in the \gls{SINR} because of the increase in the inter-cell interference caused by the addition of \gls{RIS} in the neighboring sectors. This reduction in \gls{SINR} also highlights the need for carrying out multi-cell analysis to quantify the practical gains with the inclusion of \gls{RIS}. Further, for the same reason, all the cell-edge users do not observe similar interference, and hence, the pattern of achievable gains in coupling loss and \gls{SINR} differ. Note that the performance of the random phase shifters is close to the scenario without any \gls{RIS} and the performance of the proposed \gls{RIS} with 4-bit discrete phase shifters is also significantly better than the baseline algorithms. With the 12-bit discrete phase shifters, the \gls{SINR}s of the proposed design are similar to the \gls{SINR}s of the proposed design with ideal phase shifters. Further, the gains in spectral efficiency follow a similar pattern as that of the \gls{SINR}s. As shown in Fig.~\ref{fig:SINR}, the proposed design achieves 30\% improvement in spectral efficiency as compared to the baseline algorithms.

In Fig.~\ref{fig:Theorem}, we present the comparison of achievable data rates for varying $D$ and $N$. We show that the simulated results completely align with the theoretical analysis derived in Lemma \ref{lemma:lem2} and \ref{lemma:lem3}. 
Note that with an increase in $N$, the system performance with the random phase shifts does not change and is same as that of the case with no \gls{RIS}. Further,
with the increase in $D$, the performance of the proposed design with discrete phase shifters converges to that of ideal phase shifters. 
\section{Conclusion}
\label{sec:Conclusion} 
In this letter, we have presented in detail the system-level modeling of the \gls{RIS}-assisted cellular systems. Our numerical evaluations outline that only with the system-level evaluations the actual gains achievable with \gls{RIS} can be understood.
We have proposed a novel cellular layout and \gls{RIS} organization which ensures uniform accessibility of the \gls{RIS} to all the active users in the desired sector and creates minimal interference to users in the neighboring cells. We have then proposed a beamforming design for the \gls{RIS}-assisted cellular systems and analyzed the achievable gains with the proposed design by considering ideal and discrete phase shifters. 
Through extensive system-level analysis, we have shown that the proposed design achieves significant improvements in network performance when compared to the baseline algorithms. 
In the future, we plan to evaluate the proposed design on the hardware testbed.
\bibliographystyle{ieeetran}
\bibliography{Bibfile.bib}
\end{document}